%% file: main.tex
\documentclass[11pt]{article}
\usepackage[utf8]{inputenc}
\usepackage[T1]{fontenc}
\usepackage{lmodern}
\usepackage{microtype,enumitem}

\usepackage[left=1in,right=1in,top=1in,bottom=1in]{geometry} 

\usepackage{amsmath}
\usepackage{amsfonts}
\usepackage{amsthm}
\usepackage{mathtools}

\usepackage[procnumbered,ruled,vlined,linesnumbered,algo2e]{algorithm2e}

\usepackage{comment}

\usepackage[textwidth=1.5in,textsize=footnotesize,disable]{todonotes}

\usepackage{mathtools,amsfonts,nicefrac}
\usepackage{graphicx}
\usepackage{comment}

\usepackage[hidelinks]{hyperref}

\newtheorem{theorem}{Theorem}[section]
\newtheorem{lemma}[theorem]{Lemma}

\newtheorem{fact}[theorem]{Fact}

\DontPrintSemicolon
\SetKwProg{Procedure}{Procedure}{}{}
\SetProcNameSty{textsc}
\SetFuncSty{textsc}
\SetKw{KwAnd}{and}
\SetKw{KwDo}{do}

\usepackage{wrapfig}
\newcommand{\erclogowrapped}[1]{%
\setlength\intextsep{0pt}%
\begin{wrapfigure}[3]{r}{#1*\real{1.1}}%
\includegraphics[width=#1]{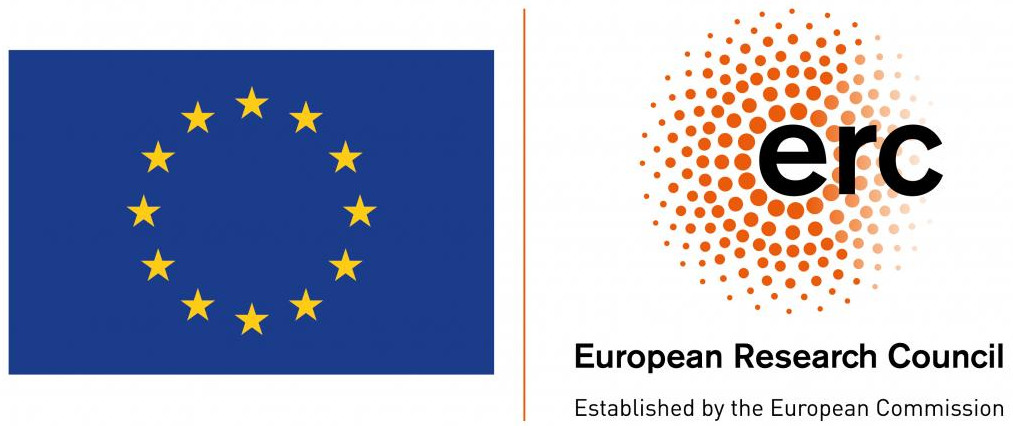}%
\end{wrapfigure}%
}

\SetKwFunction{Initialize}{Initialize}
\SetKwFunction{Delete}{Delete}
\SetKwFunction{Insert}{Insert}
\SetKwFunction{Reach}{Reach}
\SetKwFunction{UpdateTree}{UpdateTree}
\SetKwFunction{MaxFlow}{MaxFlow}

\SetKwFunction{DeleteAuxiliary}{DeleteEdgeSet}

\SetKwFunction{UpdateClusterInformation}{UpdateClusterInformation}

\newcommand{\counter}{\tau}
\newcommand{\timeMF}{t_{\textrm{static}}}
\newcommand{\IncBMF}{\textsc{IncBMF}}
\newcommand{\boundK}{\mu}
\newcommand{\totalTB}{t_{\textrm{total}}}

\newcommand{\approxF}{F}
\newcommand{\optF}{F^*}

\title{Incremental Approximate Maximum Flow in $m^{1/2+o(1)}$ update time}
\author{Gramoz Goranci\thanks{Institute for Theoretical Studies, ETH Zurich, Switzerland. This project is supported by Dr. Max R\"ossler, the Walter Haefner Foundation and the ETH Z\"urich Foundation.} 
\and Monika Henzinger\thanks{Faculty of Computer Science, University of Vienna, Austria. See funding information in the acknowledgement section.}
} 

\date{}

\begin{document}

\maketitle

\thispagestyle{empty}

\begin{abstract}
\input{abstract}
\end{abstract}


\input{introduction}

\input{preliminaries}

\input{framework}

\input{body}

\input{ack}

\small
\setlength\parskip{0cm}
\setlength\itemsep{0cm}
\bibliographystyle{plain}
\bibliography{literature_ISO}

\end{document}

%% file: abstract.tex
We show an $(1+\epsilon)$-approximation algorithm for maintaining maximum $s$-$t$ flow under $m$ edge insertions in $m^{1/2+o(1)} \epsilon^{-1/2}$ amortized update time for directed, unweighted graphs. This constitutes the first sublinear dynamic maximum flow algorithm in general sparse graphs with arbitrarily good approximation guarantee. 

%% file: introduction.tex
\section{Introduction}
The maximum flow problem is one of the cornerstone and the most studied problem in combinatorial optimization. It is often used as subroutine for solving other prominent graph problems~(e.g., Gomory-Hu Trees~\cite{GomoryH61}, Sparsest Cut~\cite{KhandekarRV09}), performing divide-and-conquer on graphs and has found several applications across many areas including computer vision, clustering and scientific computing. Designing fast maximum flow algorithms has been an active area of research for decades, with recent advances making tremendous progress towards the quest of designing a near-linear time algorithm~\cite{ChristianoKMST11, Sherman13, Madry13,KelnerLOS14,Peng16,Madry16, Sherman17, LiuS20, ChenKLPPS22}.

Recently, we have witnessed a growing interest in designing dynamic algorithms for computing maximum flows in dynamically changing graphs~\cite{ItalianoNSW11, CheungKL13, AbrahamDKKP16, Dahlgaard16, GuptaK18, GoranciRST20, CHenGHPS20}. Despite this, the current best-known algorithms either incur super-constant approximation factors~\cite{GoranciRST20,CHenGHPS20} or achieve competitive update times only for sufficiently dense graphs~\cite{AbrahamDKKP16}. Moreover, all previous works on dynamic flows are restricted to undirected graphs. From a (conditional) lower bound perspective, for any $\delta > 0$, it is known~\cite{Dahlgaard16} that no algorithm can exactly maintain maximum flow in $O(m^{1-\delta})$ amortized time per operation, even when restricted to algorithms that support edge insertions, unless the OMv conjecture~\cite{HenzingerKNS15} is false. This suggests that approximation is necessary to achieve sublinear running times for general sparse graphs. 

In this paper, we show a simple generic algorithmic framework for maintaining approximate maximum flows under edge insertions.

\begin{theorem} \label{thm: mainThm}
Let $G=(V,E)$ be an initially empty, directed, unweighted $n$-vertex graph, $s$ and $t$ be any two vertices, and $\epsilon > 0$, $\boundK \in [0,n]$ be two parameters. If there is 
\begin{enumerate}
\itemsep0em
    \item[\textup{(a)}] an incremental algorithm $\textup{\IncBMF}(G,s,t,\boundK)$ for inserting $m$ edges and maintaining $s$-$t$ maximum flow whose value is bounded by $\boundK$ in total update time $\totalTB(m,n,\boundK)$ and $q(m,n)$ query time, and
    \item[\textup{(b)}] a static algorithm for computing exact $s$-$t$ maximum flow in an $n$-vertex, $m'$-edge graph in $\timeMF(m',n)$ time, where $m' \leq m$,
\end{enumerate}
then we can design an incremental algorithm for maintaining a $(1+\epsilon)$-approximate $s$-$t$ maximum flow under $m$ edge insertions in \[ \frac{\totalTB(m,n,\boundK+1)}{m} +  \frac{\timeMF(m,n)}{\epsilon \boundK} + q(m,n) \]
amortized update time and $q(m,n)$ query time. 	
\end{theorem}

For maintaining exact maximum flows whose value is bounded by $\boundK$, we slightly adapt an incremental version of the Ford-Fulkerson~\cite{ford1962flows} algorithm, which was initially observed by Heniznger~\cite{Henzinger97} and later by Gupta and Khan~\cite{GuptaK18} (cf. Theorem~\ref{thm:incrBMF}). This gives an incremental algorithm with $O(m \boundK)$ total update time and $O(1)$ query time. The recent breakthrough result due to Chen, Kyng, Liu, Peng, Probst Gutenberg, and Sachdeva~\cite{ChenKLPPS22}~(cf. Theorem~\ref{thm: exactMaxFlow}) shows a static exact maximum flow algorithm that runs in $m^{1+o(1)}$ time. Plugging these bounds in Theorem~\ref{thm: mainThm} and choosing $\boundK = m^{1/2+o(1)}\epsilon^{-1/2}$ yields the main result of this paper, which we summarize in the theorem below.

\begin{theorem} \label{thm: mainThm2}
	Given an initially empty, directed, unweighted  graph $G=(V,E)$, any two vertices $s$ and $t$ in $V$, and any $\epsilon > 0$, there is an incremental algorithm that maintains a $(1+\epsilon)$-approximate maximum $s$-$t$ flow  in $G$ under $m$ edge insertions in $m^{1/2+o(1)} \epsilon^{-1/2}$ amortized update time. The algorithm supports queries about the value of the maintained flow in $O(1)$ time.
\end{theorem}


We believe that our approach to dynamic flows may serve as the basis for designing new fully-dynamic maximum flow algorithms with competitive approximation ratio.

%% file: preliminaries.tex
\section{Preliminaries}


In the following we settle some basic notation, as well as review definitions and algorithms for computing flows on graphs. 

\paragraph{Maximum Flow.}

Let $G=(V,E)$ be a directed, unweighted graph with $n$ vertices and $m$ edges, let $s \in V$ be a \emph{source} vertex and let $t \in V$ be a \emph{target} vertex. An \emph{flow} from $s$ to $t$ in $G$ is a function $f: E \rightarrow \mathbb{R}^{+}$ that maps each edge to a non-negative real number; the value $f(e)$ represents the amount of flow sent along $e$. A flow must satisfy the following properties: (i) for each $e \in E$, we have $f(e) \leq  1$, known as \emph{capacity constraints}, and (ii) for each $v \in V \setminus \{s,t\}$, $\sum_{(u,v) \in E} f(u,v) = \sum_{(v,u)} f(v,u)$, known as \emph{conservation constraints}. The \emph{value} of a flow $f$ is the amount of flow leaving the source $s$ minus the amount flow flow entering $s$, i.e., $v(f) = \sum_{(s,u) \in E} f(s,u) - \sum_{(u,s) \in E} f(u,s)$. In the \emph{maximum $s$-$t$ flow} problem, the goal is to find a flow $f$ with the largest $v(f)$, called the {\emph maximum $s$-$t$ flow} . Let $\optF = \max\{v(f) \mid f \text{ is a flow in } G \}$. Note that while $\optF$ is unique, there might be multiple maximum $s$-$t$ flows attaining $\optF$.

\paragraph{Residual graph and Augmenting Paths.} Given a directed, unweighted graph $G=(V,E)$, and a flow $f$ from $s$ to $t$ in $G$, we let $G_f=(V,E_f)$ be the \emph{residual graph} of $G$ with respect to $f$, where $E_f$ contains all edges of $E$, except that their direction is reversed if $f(e) = 1$. An edge whose direction in $G_f$ is reversed is referred to as a \emph{backward} edge. Otherwise, the edge is a \emph{forward} edge.  An augmenting path $P$ from $s$ to $t$ in $G$ is a  simple directed path from $s$ to $t$ in $G_f$. We next review a powerful result relating the residual graphs and optimal flows.

\begin{lemma}[\cite{ford1962flows}] \label{lem:FFoptimal}
	
	If there is no directed path from $s$ to $t$ in the residual graph $G_f$, then the flow $f$ is a maximum $s$-$t$ flow.
\end{lemma}

Another useful fact is giving in the lemma below.

\begin{fact} \label{fact:IncreaseFlow}
	If there exists a directed path $P$ from $s$ to $t$ in $G_f$,  pushing flow along $P$ in $G$ increases the value of the flow $f$ by at exactly one. 
\end{fact}

\paragraph{Exact Maximum Flow in Directed Graphs.} Our incremental approximate algorithm heavily relies on the ability to compute a maximum $s$-$t$ flow quickly. Hence, we use the current breakthrough result~\cite{ChenKLPPS22} that achieves an exact, almost-linear algorithm for the maximum flow problem on directed graphs.\footnote{The algorithm extends to graphs with polynomially bound weights, but for the purposes of this paper and simplifying the presentation, we only state the unweighted version of this result}


\begin{theorem}[\cite{ChenKLPPS22}] \label{thm: exactMaxFlow}
	For any directed, unweighted graph $G=(V,E)$ and any two vertices $s$ and $t$, there is an algorithm that computes an \emph{exact} maximum $s$-$t$ flow in $G$ in $m^{1+o(1)}$ time.
\end{theorem}

\paragraph{Approximate Maximum Flow.} 

To measure the quality of approximate maximum flows, we will use the notion of $\alpha$-approximations, which indicates that the value of the current flow solution is at least $1/\alpha$ of the optimum value. In other words, a flow $f$ is $\alpha$-approximate if $\approxF := v(f) \geq \frac{1}{\alpha} \optF$. 




%% file: framework.tex
\section{A framework for Incremental Approximate Maximum Flow}
In this section we show a simple generic algorithmic framework for maintaining a $(1+\epsilon)$-approximate maximum $s$-$t$ flow under edge insertions, i.e., prove Theorem~\ref{thm: mainThm}. Our construction is based on two important components: (i) incrementally maintaining maximum  $s$-$t$ flows whose value is upper bounded by some parameter $\boundK$ (one should think of $\boundK$ being small relative to the size of the network) and (ii) performing periodical rebuilds whenever the maximum $s$-$t$ flow of the current flow is larger than the parameter $\boundK$. 
This is a common approach in dynamic algorithms and was used e.g., by  Gupta and Peng~\cite{GuptaP13} in their dynamic algorithm for maintaining approximate matchings. 
We next further elaborate on the precise requirements of both components, and discuss how they lead to our algorithm.

To implement component (i), given a directed, unweighted graph $n$-vertex, $m$-edge graph $G$, any two vertices $s$, $t$, and a parameter $\boundK \in [0,n]$, the goal is to construct a data structure denoted by $\IncBMF(G,s,t,\boundK)$, or simply $\IncBMF$, that supports the following operations
\begin{itemize}
\itemsep0em
    \item \textsc{Initialize}$(G,s,t,\boundK)$: Initializes the data structure.
    \item \textsc{Insert}$(u,v)$: Insert the edge $(u,v)$ to $G$.
    \item \textsc{MaxFlow}$(s,t)$: Return the value of the maximum $s$-$t$ flow in the current graph $G$ if this value is smaller than $\boundK$.
\end{itemize}
Ideally, we would like that $\IncBMF(G,s,t,\boundK)$ supports edge insertions in amortized time proportional to the parameter $\boundK$, and queries in constant time. This would alone lead to an efficient incremental maximum flow algorithm whenever the current flow value is bounded by $\boundK$. 

When the maximum flow is large (i.e., component (ii)), we can make use of the \emph{stability} of maximum flow and periodically invoke a fast static algorithm. Concretely, first note that the value of the maximum $s$-$t$ flow changes by at most one per insertion. Therefore, if we have a large flow that is close to the maximum one, it will remain close to the maximum flow over a large number of updates. This naturally leads to the following simple but powerful approach: compute a flow at a certain time in the update sequence and do nothing for a certain number of updates as long as the flow is a good approximation to the maximum flow. This idea together with the data structure $\IncBMF(G,s,t,\boundK)$  yields an incremental algorithm for approximating $s$-$t$ maximum flow, which we formally describe below.


 Given a directed, unweighted graph $G=(V,E)$, any two vertices $s$, $t$, and two parameters $\epsilon>0, \boundK \in [0,n]$, our data structure maintains:
\begin{itemize}
\itemsep0em
\item a flow estimate $\approxF$ to the maximum $s$-$t$ flow $\optF$,
\item a counter $\counter$ indicating the number of operations since the last rebuild,
\item an incremental algorithm \textsc{IncBMF} for maintaining graphs with the maximum flow bounded by some parameter $(\boundK+1)$.
\end{itemize}

Initially, $G$ is an empty graph, $\approxF \gets 0$, $\counter \gets 0$, and we invoke the operation $\textsc{Initialize} (G,s,t,\boundK+1)$ of $\IncBMF$. Upon insertion of an edge $(u,v)$ to $G$, we 
query $\IncBMF$
to determine whether the $s$-$t$ maximum flow in the current graph is at most $\boundK$. If so,  we pass the edge insertion to the \textsc{\IncBMF} data structure and update $F$ accordingly. 

On the other hand, if the current maximum $s$-$t$ flow is larger than $\boundK$~(and our algorithm always correctly detects this since $\IncBMF$ run with the parameter $(\boundK+1)$ returns the correct answer),
we increment $\counter$, which counts the number of insertions since the last reset of $\counter$ that fall into this case.
If $\counter \ge \epsilon \mu$, we compute an exact maximum flow $F^*$ for the current graph from scratch using a static algorithm, update $F$ using the value of $F^*$ and
set $\counter = 0$. We call such a step a \emph{rebuild} step.
Observe that since we are in the insertions-only setting, once a maximum flow is larger than $\boundK$, it remains larger than that value. Finally, to answer a query about the maximum $s$-$t$ flow, we return $F$ as an estimate. These procedures are summarized in Algorithm~\ref{alg:incrementalAMF}.

\begin{algorithm2e}
	\caption{Incremental Approximate Maximum Flow (\textsc{IncApproxMF})}\label{alg:incrementalAMF}
	
	\Procedure{\Initialize{$ G = (V,E) $, $s$, $t$, $\epsilon$, $\boundK$}}{
		Set $E \gets 0$ and $F \gets 0$ \;
		Invoke $\textsc{IncBMF}.\textsc{Initialize}(G,s,t, \boundK+1)$
		
	}
	
	\Procedure{\Insert{$ u,v $}}{
		$E \gets E \cup \{(u,v)\}$ \;
		\If{\textup{$\textsc{IncBMF}.\textsc{MaxFlow}(s,t) \leq \boundK$}}
		{
			Invoke $\textsc{IncBMF}.\textsc{Insert}(u,v)$\;
			Set $F \gets \textsc{IncBMF}.\textsc{MaxFlow}(s,t)$
		}
		\Else{
	
		Set $\counter \gets \counter + 1$ \;
		\If{$\counter \ge \epsilon \boundK $}
		   {
		   	   Compute an $s$-$t$ maximum flow  in $G$ using a static algorithm \;
		   	   Set $\approxF$ to be the value of the flow computed in the previous step \;
		   	   Set $\counter \gets 0$\;
                   
		   }	
	}

	}
	
	\Procedure{\MaxFlow{$ s,t $}}{
		\Return $\approxF$
	}
	
\end{algorithm2e}

\begin{theorem}[Restatement of Theorem~\ref{thm: mainThm}]
Let $G=(V,E)$ be an initially empty, directed, unweighted $n$-vertex graph, $s$ and $t$ be any two vertices, and $\epsilon >0$, $\boundK \in [0,n]$ be two parameters. If there is 
\begin{enumerate}
\itemsep0em
    \item[\textup{(a)}] an incremental algorithm $\textup{\IncBMF}(G,s,t,\boundK)$ for inserting $m$ edges and maintaining $s$-$t$ maximum flow whose value is bounded by $\boundK$ in total update time $\totalTB(m,n,\boundK)$ and $q(m,n)$ query time, and
    \item[\textup{(b)}] a static algorithm for computing exact $s$-$t$ maximum flow in an $n$-vertex, $m'$-edge graph in $\timeMF(m',n)$ time, where $m' \leq m$,
\end{enumerate}
then we can design an incremental algorithm for maintaining a $(1+\epsilon)$-approximate $s$-$t$ maximum flow under $m$ edge insertions in \[ \frac{\totalTB(m,n,\boundK+1)}{m} +  \frac{\timeMF(m,n)}{\epsilon \boundK} + q(m,n) \]
amortized update time and $q(m,n)$ query time.	
\end{theorem}

\begin{proof}
	We first prove the correctness of the algorithm. Let $G$ be the current graph and let $\approxF$ be the estimate maintained by the algorithm to the value of the maximum $s$-$t$ flow $\optF$ in $G$. 
 We will show that $\approxF$ is an $(1+\epsilon)$-approximation to $\optF$. 
 To this end, we distinguish the following three cases. (1) If $\optF \leq \boundK$, then by assumption of the theorem, the data structure \IncBMF ensures that $\approxF = \optF$ and thus our claim trivially holds. 
(2) If  $\optF = \boundK +1$, the  call 
\textup{$\textsc{IncBMF}.\textsc{MaxFlow}(s,t)$} returns 
the value $\boundK +1$ and, thus, the algorithm reaches the else-case for the first time (here we slightly abuse the notation and denote this as a \emph{rebuild} step).

 (3) If $\optF > \boundK + 1$, then this is not the first time that the algorithm reaches the else-case and, thus, there was a prior rebuild. Note that $\approxF$ corresponded to the value of some $s$-$t$ maximum flow at the last prior rebuild. This in turn implies that $\approxF$ must be larger than $\boundK$.
 Let $\optF_{0}$ be the value of the maximum $s$-$t$ flow of the graph at that rebuild. Since each edge insertions can increase the value of the maximum flow by at most $1$ and we recompute a new maximum flow every $\epsilon \boundK $ insertions, we have that $\optF \leq \optF_{0} + \epsilon \boundK$. Since $\optF_0 > \boundK $ and $F = \optF_{0}\ge 1$, bringing these together yields:
	\begin{align*}
	\frac{\optF}{\approxF} & \leq \frac{\optF_0 + \epsilon \boundK}{F} \leq \frac{(1+\epsilon ) \optF_0}{\optF_0} \leq 1+\epsilon,
	\end{align*}
	which proves our claimed approximation guarantee.
	
	We next study the running time. Note that our algorithm passes the edge insertions to the incremental algorithm \textsc{IncBMF} (invoked with the parameter $\boundK+1$) 
 only if the value of the maximum flow in the current graph is bounded by $\boundK$. Hence, by the theorem assumption, the total update time to handle these insertions is $\totalTB(m,n,\boundK+1) + mq(m,n)$. Amortizing the latter over $m$ insertions gives an amortize cost of $\totalTB(m,n,\boundK+1)/m + q(m,n)$, which in turn gives the first and the third term of our claimed running time guarantee.

 It remains to analyse the cost of periodical rebuilds. Note that if the current maximum flow value is larger than $\boundK$, our algorithm updates the estimate $F$ every $\epsilon \boundK$ operations. By assumption of the theorem, the time to compute an exact maximum flow is $\timeMF(m',n) \leq \timeMF(m,n)$ as $m' \leq m$. Charging this time over $\epsilon \boundK$ insertions, yields an amortized cost of $\timeMF(m,n)/(\epsilon \boundK)$, which in turn gives the second term our of claimed runtime guarantee and completes the proof of the theorem.
\end{proof}


%% file: body.tex
\section{Incremental Bounded Maximum Flow}
\label{sec: incrementalBMaxFlow}

In this section we give an incremental algorithm for exactly maintaining the maximum flow as long as its value is bounded by a predefined parameter $\boundK$. The algorithm applies to directed, unweighted graphs, was initially observed by Henzinger~\cite{Henzinger97} and later by Gupta and Khan~\cite{GuptaK18}, and can be thought of as an incremental version of the celebrated Ford-Fulkserson algorithm~\cite{ford1962flows}. We review it below and slightly adapt it for our purposes.

Henzinger~\cite{Henzinger97} showed how to incrementally maintain maximum $s$-$t$ flow in $O(\optF)$ amortized update time, where $\optF$ is the value of the maximum flow in the final graph. As there are graphs where $\optF = \Omega(n)$, their running time guarantee is competitive only when $\optF$ is small, e.g., sub-linear on the size of the graph. We next show to slightly adapt their algorithm so that it maintains a maximum flow as long as its value is bounded by a parameter $\boundK$.

The key observation behind this algorithm is that the insertions of an (unit-capacitated) edge can only increase the maximum flow value by at most $1$. To check whether this value has increased, they use Lemma~\ref{lem:FFoptimal} as a certificate, i.e., one determines whether the insertion of the (forward) edge in the residual graph $G_f$ creates a directed path from $s$ to $t$ in $G_f$. A naive way to determine this is to run a graph search algorithm on $G_f$ after each insertions, which requires $\Omega(m)$ for a single update and is thus prohibitively expensive for our purposes. However, one can exploit a data structure due to Italiano~\cite{Italiano86} for incrementally maintaining single source reachability information from a source $s$ which requires $O(m)$ total update time for handling $m$ insertions. Let us briefly review this data structure before presenting the incremental algorithm.

\paragraph*{Incremental Single Source Reachability.}
In the incremental single source reachability problem, given an (initially empty) directed, unweighted graph $G=(V,E)$ and a distinguished vertex $s$, the goal is to construct a data structure $\textsc{IncSSR}$ that supports the following operations: (i) \textsc{Initialize}$(G,s)$: initialize the data structure in $G$ with source $s$, (ii) $\textsc{Insert}(u,v)$: insert the edge $(u,v)$ in $G$, and (iii) \textsc{Reach}$(u)$: return \texttt{True} if $u$ is reachable from $s$, and \texttt{False} otherwise. 

Italiano~\cite{Italiano86} observed that an incremental version of graph search leads to an efficient incremental $\textsc{IncSSR}$ data structure. The main idea is to maintain a reachability tree $T$ from $s$. Initially, the tree is initialized to $\{s\}$. Upon insertion of an edge $(u,v)$ to $G$, we need to update $T$ iff $u \in T$ and $v \not \in T$. If this is the case, we add $(u,v)$ to $T$ and make the $v$ the child of $u$. Moreover, the algorithm examines all outgoing neighbors $w$ incident to $v$, and if $w \not \in T$, processes the edge $(v,w)$ recursively using the same procedure. To answer queries, we return \texttt{True} if $u \in T$, and \texttt{False} otherwise. These procedures are summarized in Algorithm~\ref{alg:incrementalSSR}.

\begin{algorithm2e}[h]
	\caption{Incremental Single Source Reachability (\textsc{IncSSR})}\label{alg:incrementalSSR}
	
	\Procedure{\Initialize{$ G = (V,E) $, $s$}}{
	 	Set $T \gets \{s\}$ and $E \gets \emptyset$
	}
	
	\Procedure{\Insert{$ u,v $}}{
		$E \gets E \cup \{(u,v)\} $ \; 
		\UpdateTree{$u,v$}
	}

	\Procedure{\UpdateTree{$ u,v $}}{
	\If{$u \in T$ and $v \not \in T$}
		{
			Make $v$ a child of $u$ in $T$ \;
			\ForEach{$(v,w) \in E$}
			{
				\UpdateTree{$v,w$}
			}
		}
}
	
	\Procedure{\Reach{$ u $}}{
		\If{$u \in T$}
		{
			\Return \texttt{True}
		}
	    \Else{ \Return \texttt{False}}
	}
	
\end{algorithm2e}

The correctness of the data structure immediately follows by construction as we always maintain a correct rechability tree $T$ from $s$ for the current graph. For the running time, note that the total time over all insertions is $O(m)$ as each edge is processed at most $O(1)$ times, once when it is inserted into the graph and once when it is added to $T$.

\begin{lemma}[\cite{Italiano86}] \label{lem:incrementalSSR}
	Given an initially empty directed, unweighted graph $G=(V,E)$ and a source vertex $s$, the incremental algorithm $\textup{\textsc{IncSSR}}$ maintains reachability information from $s$ to every other node in $V$ while supporting insertions in $O(1)$ amortized update time and queries in $O(1)$ in worst-case time. 
\end{lemma}

\paragraph{The Algorithm.} We now have all the necessary tools to present an incremental algorithm maintaining the maximum flow whose value is bounded by $\boundK$. Let $\optF$ denote the maximum $s$-$t$ flow value on the current graph. 

Initially, $G$ and the residual graph $G_f$ are empty graphs, $\optF \gets 0$ and $f(e) \gets 0$ for each $e \in E$. The algorithm proceeds in $\boundK$ rounds, where a round ends when the value of the current maximum flow increases by one.  Each round starts by initializing an incremental single source reachability data structure $\textsc{IncrSSR}$ from the source $s$~(Lemma~\ref{lem:incrementalSSR}) on the residual graph $G_f$. Upon an edge insertion $(u,v)$ to $G$, we pass the directed edge $(u,v)$ to the data structure $\textsc{IncrSSR}$ and test whether $t$ is reachable from $s$ using this data-structure. If the latter holds, then we find a simple directed $s$-$t$ path $P$ in $G_f$, which in turn serves as an augmenting path for $G$. We then send one unit of flow along the path $P$ in $G$ and update the current flow and its value accordingly. To answer a query about the maximum flow between $s$ and $t$, we simply return $\optF$. These procedures are summarized in Algorithm~\ref{alg:incrementalBMF}.

\vspace{0.5cm}
\begin{algorithm2e}[H]
	\caption{Incremental Bounded Maximum Flow (\textsc{IncBMF})}\label{alg:incrementalBMF}
	
	\Procedure{\Initialize{$ G = (V,E) $, $s$, $t$, $\boundK$}}{
		Set $E \gets \emptyset$ \;
        Set  $f(e) \gets 0$ for each $e \in E$,  $G_f \gets (V, E)$ and $\optF \gets 0$ \;
		Invoke $\textsc{IncSSR}.\textsc{Initialize}(G_f,s)$
		
	}
	
	\Procedure{\Insert{$ u,v $}}{
            \If{$\optF \leq \boundK$}
            {
		$E \gets E \cup \{(u,v)\}$ \;
		Invoke  $\textsc{IncSSR}.\textsc{Insert}(u,v)$ \; 
		\If{$\textup{\textsc{IncSSR}.\textsc{Reach}}(t)$}
		   {
			Find a simple directed $s$-$t$ path $P$ in $G_f$ \;
			Augment $f$ along the path $P$ in $G$ and let $f'$ be the resulting flow  \;
			Set $f \gets f'$ and $G_f \gets G_{f'}$ \;
			Set $\optF \gets \optF + 1$	\;
			Invoke $\textsc{IncSSR}.\textsc{Initialize}(G_f,s)$
		   }
            }

	}
	
	\Procedure{\MaxFlow{$ s,t $}}{
		\Return $\optF$
	}
\end{algorithm2e}
\vspace{0.5cm}

The correctness of this algorithm is immediate by Lemma~\ref{lem:FFoptimal}, which correctly tells us when to increase the value of the maximum flow, and Fact~\ref{fact:IncreaseFlow}, which asserts that the sending one unit of flow along an augmenting path increase the value of the flow by exactly one.

For the running time, note that each round requires $O(m)$ total time. As there are exactly $\boundK$ stages, we get a total update time of $O(m \boundK)$.

\begin{theorem} \label{thm:incrBMF}
	Given an initially empty directed, unweighted graph $G=(V,E)$ with $n$ vertices, any two vertices $s$ and $t$, and a parameter $\boundK \in [0,n]$, the algorithm \textup{\textsc{IncBMF}}$(G,s,t,\boundK)$ exactly maintains, under $m$ edge insertions, the maximum $s$-$t$ flow in $G$ whose value is bounded by $\boundK$ in $O(m \boundK)$ total update time and $O(1)$ query time.
\end{theorem}

%% file: ack.tex
\section*{Acknowledgement}
We thank Richard Peng, Thatchaphol Saranurak, Sebastian Forster and Sushant Sachdeva for helpful discussions.

\erclogowrapped{5\baselineskip} \emph{M. Henzinger:} This project has received funding from the European Research Council (ERC) under the European Union's Horizon 2020 research and innovation programme (Grant agreement No. 101019564 ``The Design of Modern Fully Dynamic Data Structures (MoDynStruct)'' and from the Austrian Science Fund (FWF) project ``Fast Algorithms for a Reactive Network Layer (ReactNet)'', P~33775-N, with additional funding from the \textit{netidee SCIENCE Stiftung}, 2020--2024.